%% file: main.tex
\documentclass{llncs}

\input{preamble.tex}

\newif \ifextended
\usepackage{makeidx}  
\begin{document}
\extendedtrue
\frontmatter          
\pagestyle{plain}
\mainmatter              
\title{Shield Synthesis:}
\subtitle{Runtime Enforcement for Reactive Systems%
\thanks{This work was supported in part by the Austrian Science Fund
  (FWF) through the research network RiSE (S11406-N23) and by the
  European Commission through project STANCE (317753). Chao Wang is
  supported by the National Science Foundation grant CNS-1128903.}}
\titlerunning{Shield Synthesis}  
%
\author{Roderick Bloem\inst{1} \and
        Bettina K\"onighofer\inst{1} \and
        Robert K\"onighofer\inst{1} \and
        Chao Wang\inst{2}}
\authorrunning{Roderick Bloem et al.} 
%
%
\institute{
   $^1$ IAIK, Graz University of Technology,  Austria\\
   $^2$ Department of ECE, Virginia Tech, Blacksburg, VA 24061, USA
          }

\maketitle              

\begin{abstract}
Scalability issues may prevent users from verifying critical properties 
of a complex hardware design.  In this situation, we propose to 
synthesize a ``safety shield'' that is attached to the design to enforce 
the properties at run time.  \emph{Shield synthesis} can succeed where 
model checking and reactive synthesis fail, because it only considers a 
small set of critical properties, as opposed to the complex design, or 
the complete specification in the case of reactive synthesis. The shield 
continuously monitors the input/output of the design and corrects its 
erroneous output only if necessary, and as little as possible, so other 
non-critical properties are likely to be retained.  Although runtime 
enforcement has been studied in other domains such as action systems, 
reactive systems pose unique challenges where the shield must act 
without delay.  We thus present the first shield synthesis solution for 
reactive hardware systems and report our experimental results.
\ifextended
This is an extended version of~\cite{conference}, featuring an 
additional appendix.
\fi
\end{abstract}

\section{Introduction}

Model checking~\cite{Clarke81,Quiell81} can formally verify that a 
design satisfies a temporal logic specification.  Yet, due to 
scalability problems, it may be infeasible to prove all critical 
properties of a complex design.  Reactive 
synthesis~\cite{Pnueli89,BloemJPPS12} is even more ambitious since it 
aims to generate a provably correct design from a given specification. 
In addition to scalability problems, reactive synthesis has the drawback 
of requiring a complete specification, which describes every aspect of 
the desired design.  However, writing a complete specification can 
sometimes be as hard as implementing the design itself.

\begin{wrapfigure}[7]{r}{0.45\textwidth}
\vspace{-0.6cm}
\includegraphics[width=0.45\textwidth]{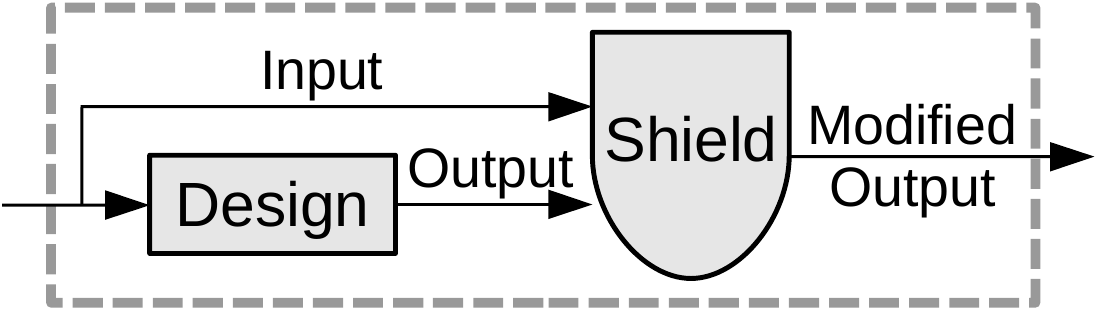}
\caption{Attaching a safety shield.\\}
\label{fig:attach_shield}
\end{wrapfigure}

We propose \emph{shield synthesis} as a way to complement model checking 
and reactive synthesis.  Our goal is to enforce a small set of critical 
properties at runtime even if these properties may occasionally be 
violated by the design.  Imagine a complex design and a set of 
properties that cannot be proved due to scalability issues or other 
reasons (e.g., third-party IP cores). In this setting, we are in good 
faith that the properties hold but we need to have certainty.  We would 
like to automatically construct a component, called the \emph{shield}, 
and attach it to the design as illustrated in 
Fig.~\ref{fig:attach_shield}. The shield monitors the input/output of 
the design and corrects the erroneous output values instantaneously, but 
only if necessary and as little as possible.

The shield ensures both \emph{correctness} and \emph{minimum 
interference}.  By correctness, we mean that the properties must be 
satisfied by the combined system, even if they are occasionally violated 
by the design.  By minimum interference, we mean that the output of the 
shield deviates from the output of the design only if necessary, and the 
deviation is kept minimum.  The latter requirement is important because 
we want the design to retain other (non-critical) behaviors that are not 
captured by the given set of properties.  We argue that shield synthesis 
can succeed even if model checking and reactive synthesis fail due to 
scalability issues, because it has to enforce only a small set of 
critical properties, regardless of the implementation details of a 
complex design. 

This paper makes two contributions. First, we define a general framework 
for solving the shield synthesis problem for reactive hardware systems. 
Second, we propose a new synthesis method, which automatically 
constructs a shield from a set of safety properties. To minimize 
deviations of the shield from the original design, we propose a new 
notion called \emph{$k$-stabilization}:  When the design arrives at a 
state where a property violation becomes unavoidable for some possible 
future inputs, the shield is allowed to deviate for at most $k$ 
consecutive steps.  If a second violation happens during the $k$-step 
recovery phase, the shield enters a \emph{fail-safe} mode where it only 
enforces correctness, but no longer minimizes the deviation. We show 
that the $k$-stabilizing shield synthesis problem can be reduced to 
\emph{safety games}~\cite{Mazala01}. Following this approach, we present a 
proof-of-concept implementation and give the first experimental 
results.

Our work on shield synthesis can complement model checking by enforcing 
any property that cannot be formally proved on a complex design.  There 
can be more applications.  For example, we may not trust third-party IP 
components in our system, but in this case, model checking cannot be 
used because we do not have the source code. Nevertheless, a shield can 
enforce critical interface assumptions of these IP components at run 
time.
Shields may also be used to simplify certification.  Instead of 
certifying a complex design against critical requirements, we can 
synthesize a shield to enforce them, regardless of the behavior of the 
design.  Then, we only need to certify this shield, or the synthesis 
procedure, against the critical requirements.
Finally, shield synthesis is a promising new direction for synthesis in 
general, because it has the strengths of reactive synthesis while 
avoiding its weaknesses --- the set of critical properties can be small 
and relatively easy to specify --- which implies scalability and 
usability.

\noindent \textbf{Related work.}  
Shield synthesis is different from recent works on reactive 
synthesis~\cite{Pnueli89,BloemJPPS12,EhlersT14}, which revisited 
Church's problem~\cite{Church63,Buchi69,Rabin72} on constructing correct 
systems from logical specifications. Although there are some works on 
runtime enforcement of properties in other 
domains~\cite{Schneider00,LigattiBW09,FalconeFM12}, they are based on 
assumptions that do not work for reactive hardware systems. 
Specifically, Schneider~\cite{Schneider00} proposed a method that simply 
halts a program in case of a violation.  Ligatti et 
al.~\cite{LigattiBW09} used edit automata to suppress or insert actions, 
and Falcone et al.~\cite{FalconeFM12} proposed to buffer actions and 
dump them once the execution is shown to be safe. None of these 
approaches is appropriate for reactive systems where the shield must act 
upon erroneous outputs on-the-fly, i.e., without delay and without 
knowing what future inputs/outputs are.  In particular, our shield 
cannot insert or delete time steps, and cannot halt in the case 
of a violation.

Methodologically, our new synthesis algorithm builds upon the existing 
work on synthesis of robust systems~\cite{BloemCGHHJKK14}, which aims to 
generate a complete design that satisfies as many properties of a 
specification as possible if assumptions are violated.  However, our 
goal is to synthesize a shield component $S$, which can be attached to 
any design $D$, to ensure that the combined system $(S\comp D)$ 
satisfies a given set of critical properties.  Our method aims at 
minimizing the ratio between shield deviations and property violations 
by the design, but achieves it by solving pure safety games.  
Furthermore, the synthesis method in \cite{BloemCGHHJKK14} uses 
heuristics and user input to decide from which state to continue 
monitoring the environmental behavior, whereas we use a subset 
construction to capture all possibilities to avoid unjust verdicts by 
the shield.  We use the notion of $k$-stabilization to quantify the 
shield deviation from the design, which has similarities to Ehlers and 
Topcu's notion of $k$-resilience in robust synthesis~\cite{EhlersT14} 
for GR(1) specifications~\cite{BloemJPPS12}.  However, the context of 
our work is different, and our $k$-stabilization limits the length of 
the recovery period instead of tolerating bursts of up to $k$ glitches.

\noindent \textbf{Outline.}  
The remainder of this paper is organized as follows.  We illustrate the 
technical challenges and our solutions in Section~\ref{sec:ex} using an 
example.  Then, we establish notation in Section~\ref{sec:prelim}.  We 
formalize the problem in a general framework for shield synthesis in 
Section~\ref{sec:frame}, and present our new method in 
Section~\ref{sec:sol}.  We present our experimental results in 
Section~\ref{sec:exp} and, finally, give our conclusions in 
Section~\ref{sec:conc}.

\section{Motivation}
\label{sec:ex}

In this section, we illustrate the challenges associated with shield 
synthesis and then briefly explain our solution using an example.  
We start with a traffic light controller that handles a single crossing 
between a highway and a farm road.  There are red (\textsf{r}) or green 
(\textsf{g}) lights for both roads.  An input signal, denoted 
$\textsf{p}\in\{0,1\}$, indicates whether an emergency vehicle is 
approaching.  The controller takes $\textsf{p}$ as input and returns 
$\textsf{h,f}$ as output.  Here, $\textsf{h}\in\{r,g\}$ and 
$\textsf{f}\in\{r,g\}$ are the lights for highway and farm road, 
respectively.
Although the traffic light controller interface is simple, the actual 
implementation can be complex.  For example, the controller may have to 
be synchronized with other traffic lights, and it can have input sensors 
for cars, buttons for pedestrians, and sophisticated algorithms to 
optimize traffic throughput and latency based on all sensors, the time 
of the day, and even the weather.  As a result, the actual design may 
become too complex to be formally verified.  Nevertheless, we want to 
ensure that a handful of safety critical properties are satisfied with 
certainty.  Below are three example properties:
\begin{enumerate}
\item The output \textsf{gg} --- meaning that both roads have green
  lights --- is never allowed.
\item If an emergency vehicle is approaching ($\textsf{p}=1$), the
  output must be \textsf{rr}.
\item The output cannot change from \textsf{gr} to \textsf{rg}, or
  vice versa, without passing \textsf{rr}.
\end{enumerate}
We want to synthesize a safety shield that can be attached to any 
implementation of this traffic light controller, to enforce these 
properties at run time.

In a first exercise, we only consider enforcing Properties 1 and 2.  
These are simple invariance properties without any temporal aspects.  
Such properties can be represented by a truth table as shown in 
Fig.~\ref{fig:traffic_spec2} (left). We use 0 to encode \textsf{r}, and 
1 to encode \textsf{g}.  Forbidden behavior is marked in bold red.  The 
shield must ensure both correctness and
minimum interference. That is, it should only change the output for red
entries. 
\begin{wrapfigure}[9]{r}{0.52\textwidth}
\vspace{-0.7cm}
\begin{minipage}{0.21\textwidth}
\centering
{\scriptsize \tt
\begin{tabular}{ccc|cc}\hline
    \textsf{p}~ & \textsf{h}~ & \textsf{f}~ & ~\textsf{h'} & 
~\textsf{f'} \\\hline
    0~ & 0~ &  0~ &  ~0 &  0\\ 
    0~ & 0~ &  1~ &  ~0 &  1\\
    0~ & 1~ &  0~ &  ~1 &  0\\
    0~ & \textcolor{darkred}{\bf 1}~ &\textcolor{darkred}{\bf 1}~ & 
~\textcolor{blue}{\bf 1} & \textcolor{blue}{\bf 0}\\
\hline
    1~ & 0~ &  0~ &  ~0 &  0\\ 
    1~ & \textcolor{darkred}{\bf 0}~ &\textcolor{darkred}{\bf 1}~ & 
~\textcolor{blue}{\bf 0} & \textcolor{blue}{\bf 0}\\
    1~ & \textcolor{darkred}{\bf 1}~ &\textcolor{darkred}{\bf 0}~ & 
~\textcolor{blue}{\bf 0} & \textcolor{blue}{\bf 0}\\
    1~ & \textcolor{darkred}{\bf 1}~ &\textcolor{darkred}{\bf 1}~ & 
~\textcolor{blue}{\bf 0} & \textcolor{blue}{\bf 0}\\
\hline
\end{tabular}
}
\end{minipage}
$\Rightarrow$
\begin{minipage}{.27\textwidth}
\centering
{\scriptsize \tt
$\mathsf{h' = \neg p \wedge h}$\\
$\mathsf{f' = \neg p \wedge \neg h \wedge f}$\\
}
\verb' '\\
\scalebox{.5}{\input{shield_spec2.pstex_t}}
\end{minipage}
\caption{Enforcing Properties 1 and 2.}
\label{fig:traffic_spec2}
\end{wrapfigure}
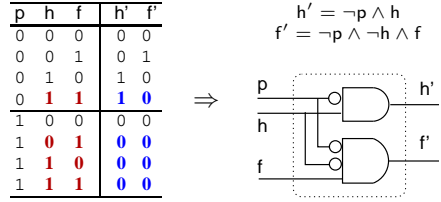 
In particular, it should not ignore the design and hard-wire 
the output to \textsf{rr}.  When $\mathsf{p=1}$ but the output is not 
$\mathsf{rr}$, the shield must correct the output to $\mathsf{rr}$.  
When $\mathsf{p=0}$ but the output is \textsf{gg}, the shield must turn 
the original output \textsf{gg} into either \textsf{rg}, \textsf{gr}, or 
\textsf{rr}.  Assume that \textsf{gr} is chosen.  As illustrated in 
Fig.~\ref{fig:traffic_spec2} (right),  we can construct the transition 
functions $h'=\neg p\wedge h$ and $f'=\neg p \wedge \neg h \wedge f$, as 
well as the shield circuit accordingly.

Next, we consider enforcing Properties 1--3 together.  Property 3 
brings in a temporal aspect, so a simple truth table does not suffice 
any more. Instead, we express the properties by an automaton, which is 
shown in Fig.~\ref{fig:traffic_spec}. Edges are labeled by values of 
$\textsf{phf}$, where $\mathsf{p}\in\{0,1\}$ is the controller's input 
and $\mathsf{h,f}$ are outputs for highway and farm road.  
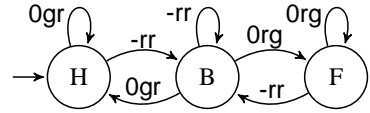
\begin{wrapfigure}[6]{r}{0.45\textwidth}
\vspace{-0.9cm}
\scalebox{0.95}{\input{figs/traffic_spec.tex}}
\caption{Traffic light specification.}
\label{fig:traffic_spec}
\end{wrapfigure}
There are three non-error states: H denotes the state where highway has 
the green light, F denotes the state where farm road has the green 
light, and B denotes the state where both have red lights.  There is 
also an error state, which is not shown. Missing edges lead to this 
error state, denoting forbidden situations, e.g., \textsf{1gr} is not 
allowed in state H.
Although the automaton still is not a complete specification, the 
corresponding shield can prevent catastrophic failures.  By 
automatically generating a small shield as shown in 
Fig.~\ref{fig:attach_shield}, our approach has the advantage of 
combining the functionality and performance of the aggressively 
optimized implementation with guaranteed safety.

While the shield for Property 1 and 2 could be realized by purely 
combinational logic, this is not possible for the specification in 
Fig.~\ref{fig:traffic_spec}.  The reason is the temporal aspect brought 
in by Property 3.
For example, if we are in state F and observe \textsf{0gg}, which is not 
allowed, the shield has to make a correction in the output signals to 
avoid the violation.  There are two options: changing the output from 
\textsf{gg} to either \textsf{rg} or \textsf{rr}.  However, this fix may 
result in the next state being either B or F.  The question is, without 
knowing what the future inputs/outputs are, how do we decide \emph{from 
which state the shield should continue to monitor} the behavior of the 
design in order to best detect and correct future violations?  If the 
shield makes a wrong guess now, it may lead to a suboptimal 
implementation that causes unnecessarily large deviation in the future.

To solve this problem, we adopt the most conservative approach. That is, 
we assume that the design $\design$ meant to give one of the allowed 
outputs, so either \textsf{rr} or \textsf{rg}. Thus, our shield 
continues to monitor the design from both F and B.  Technically, this is 
achieved by a form of subset construction (see Sec.~\ref{sec:k-stab}), 
which tracks all possibilities for now, and then gradually refines its 
knowledge with future observations.  For example, if the next 
observation is \textsf{0gr}, we assume that the design $\design$ meant 
\textsf{rr} earlier, and so it must be in B and traverse to H.
If it were in F, we could only have explained \textsf{0gr} by assuming a 
second violation, which is less optimistic than we would like to be.  
In this work, we assume that a second violation occurs only if an 
observation is inconsistent with \emph{all} states that it could 
possibly be in.  For example, if the next observation is not 
\textsf{0gr} but \textsf{1rg}, which is neither allowed in F nor in B, 
we know that a second violation occurs.  Yet, after observing 
\textsf{1rg}, we can be sure that we have reached the state B, because 
starting from both F and B, with input $\textsf{p}=1$, the only allowed 
output is \textsf{rr}, and the next state is always B.
In this sense, our construction implements an ``innocent until proved 
guilty'' philosophy, which is key to satisfy the \emph{minimum 
interference} requirement.

To bound the deviation of the shield when a property violation becomes 
unavoidable, we require the shield to deviate for at most $k$ 
consecutive steps after the initial violation.  We shall formalize this 
notion of \emph{$k$-stabilization} in subsequent sections and present 
our synthesis algorithm.
For the safety specification in Fig.~\ref{fig:traffic_spec}, our method 
would reduce the shield synthesis problem into a set of \emph{safety 
games}, which are then solved using standard techniques 
(cf.~\cite{Mazala01}).  We shall present the synthesis results in 
Section~\ref{sec:exp}.

\section{Preliminaries}
\label{sec:prelim}

We denote the Boolean domain by $\B=\{\true,\false\}$, denote the set of 
natural numbers by $\N$, and abbreviate $\N\cup\{\infty\}$ by 
$\N^\infty$.
We consider a reactive system with a finite set 
$\din=\{i_1,\ldots,i_m\}$ of Boolean inputs and a finite set 
$\dout=\{o_1,\ldots,o_n\}$ of Boolean outputs.  The input alphabet is 
$\dinalph=2^\din$, the output alphabet is $\doutalph=2^O$, and 
$\dalph=\dinalph \times \doutalph$. The set of finite (infinite) words 
over $\dalph$ is denoted by $\dalph^*$ ($\dalph^\omega$), and 
$\dalph^{*,\omega} = \dalph^* \cup \dalph^\omega$.  We will also refer 
to words as \emph{(execution) traces}.  We write $|\dtrace|$ for the 
length of a trace $\dtrace\in \dalph^{*,\omega}$. For $\dintrace = x_0 
x_1 \ldots \in \dinalph^\omega$ and $\douttrace = y_0 y_1 \ldots \in 
\doutalph^\omega$, we write $\dintrace || \douttrace$ for the 
composition $(x_0,y_0) (x_1,y_1) \ldots \in \dalph^\omega$. A set $\lang 
\subseteq  \dalph^\omega$ of infinite words is called a \emph{language}. 
We denote the set of all languages as $\langset = 2^{\dalph^\omega}$.

\noindent
\textbf{Reactive Systems.}
A \emph{reactive system} $\design = (\states, \init, \dinalph, 
\doutalph, \delta, \lambda)$ is a Mealy machine, where $\states$ is a 
finite set of states, $\init\in \states$ is the initial state, $\delta: 
\states \times \dinalph \rightarrow \states$ is a complete transition 
function, and $\lambda: \states \times \dinalph \rightarrow \doutalph$ 
is a complete output function.  Given the input trace $\dintrace = x_0 
x_1 \ldots \in \dinalph^\omega$, the system $\design$ produces the 
output trace $\douttrace = \design(\dintrace) = \lambda(q_0, x_0) 
\lambda(q_1, x_1) \ldots \in \doutalph^\omega$, where $q_{i+1} = 
\delta(q_i, x_i)$ for all $i \ge 0$.  The set of words produced by 
$\design$ is denoted $\lang(\design) = \{\dintrace || \douttrace \in 
\dalph^\omega \mid \design(\dintrace) = \douttrace\}$.  We also refer to 
a reactive system $\design$ as a \emph{(hardware) design}.

Let $\design = (\states, \init, \dinalph, \doutalph, \delta, \lambda)$ 
and $\design' = (\states', \init', \dalph, \doutalph, \delta', 
\lambda')$ be reactive systems.  Their serial composition is constructed 
by feeding the input and output of $\design$ to $\design'$ as input.  We 
use $\design \comp \design'$ to denote such a composition 
$(\hat{\states}, \hat{\init}, \dinalph, \doutalph, \hat{\delta}, 
\hat{\lambda})$, where
$\hat{\states} = \states \times \states'$, 
$\hat{\init} = (\init, \init')$,
$\hat{\delta}((q,q'),\dinletter) = (\delta(q,\dinletter),
                   \delta'(q',(\dinletter,\lambda(q,\dinletter))))$, and
$\hat{\lambda}((q,q'),\dinletter) =
                     \lambda'(q',(\dinletter,\lambda(q,\dinletter)))$.   

\noindent
\textbf{Specifications.}
A \emph{specification} $\spec$ defines a set $\lang(\spec) \subseteq 
\dalph^\omega$ of allowed traces.  A specification $\spec$ is 
\emph{realizable} if there exists a design $\design$ that realizes it. 
$\design$ \emph{realizes} $\spec$, written $\design \models \spec$, iff 
$\lang(\design) \subseteq \lang(\spec)$.
We assume that $\spec$ is a (potentially incomplete) set of 
\emph{properties} $\{\spec_1,\ldots,\spec_l\}$ such that $\lang(\spec) = 
\bigcap_i \lang(\spec_i)$, and a design satisfies $\spec$ iff it 
satisfies all its properties.
In this work, we are concerned with a \emph{safety} specification 
$\spec^s$, which is represented by an automaton $\spec^s = (\states, 
\init, \dalph, \delta, F)$, where $\dalph = \dinalph\cup\doutalph$, 
$\delta : \states \times \dalph \rightarrow \states$, and $F\subseteq 
\states$ is a set of safe states.  The \emph{run} induced by trace 
$\dtrace = \dletter_0 \dletter_1 \ldots \in \dalph^\omega$ is the state 
sequence $\overline{q} = q_0 q_1 \ldots $ such that $q_{i+1} = 
\delta(q_i, \dletter_i)$.  Trace $\dtrace$ (of a design $\design$) 
\emph{satisfies} $\spec^s$ if the induced run visits only the safe 
states, i.e., $\forall i\geq 0 \scope q_i \in F$.  The \emph{language} 
$\lang(\spec^s)$ is the set of all traces satisfying $\spec^s$.

\noindent
\textbf{Games.}
A (2-player, alternating) \emph{game} is a tuple $\game = (\gstates, 
\ginit, \dinalph, \doutalph, \delta, \win)$, where $\gstates$ is a 
finite set of game states, $\ginit \in \gstates$ is the initial state, 
$\delta: \gstates \times \dinalph \times \doutalph \rightarrow \gstates$ 
is a complete transition function, and $\win: \gstates^\omega 
\rightarrow \B$ is a winning condition.  The game is played by two 
players: the system and the environment.  In every state $g\in \gstates$ 
(starting with $\ginit$), the environment first chooses an input letter 
$\dinletter \in \dinalph$, and then the system chooses some output 
letter $\doutletter \in \doutalph$. This defines the next state $g' = 
\delta(g,\dinletter, \doutletter)$, and so on. The resulting (infinite) 
sequence $\overline{g} = g_0 g_1 \ldots$ of game states is called a 
\emph{play}.  A play is \emph{won} by the system iff 
$\win(\overline{g})$ is $\true$.

A \emph{safety game} defines $\win$ via a set $F^g\subseteq \gstates$ of 
safe states: $\win(g_0 g_1 \ldots)$ is $\true$ iff $\forall i \geq 0 
\scope g_i \in F^g$, i.e., if only safe states are visited.
A (memoryless) \emph{strategy} for the system is a function $\rho: 
\gstates \times \dinalph \rightarrow \doutalph$. A strategy is 
\emph{winning} for the system if all plays $\overline{g}$ that can be 
constructed when defining the outputs using the strategy satisfy 
$\win(\overline{g})$. The \emph{winning region} is the set of states 
from which a winning strategy exists. We will use safety games to 
synthesize a shield, which implements the winning strategy in a new 
reactive system $\shield = (\gstates, \init, \dinalph, \doutalph, 
\delta', \rho)$ with $\delta'(g,\dinletter) = 
\delta(g,\dinletter,\rho(g,\dinletter))$.

\section{The Shield Synthesis Framework} 
\label{sec:frame}

We define a general framework for shield synthesis in this section 
before presenting a concrete realization of this framework in the next 
section.

\begin{definition}[Shield]
Let $\design = (\states, \init, \dinalph, \doutalph, \delta, \lambda)$ 
be a design, $\spec$ be a set of properties, and $\specv\subseteq \spec$ 
be a valid subset such that $\design \models \specv$.  A reactive system 
$\shield = (\states', \init', \dalph, \doutalph, \delta', \lambda')$ is 
a \emph{shield} of $\design$ with respect to $(\spec\setminus\specv)$ 
iff $(\design \comp \shield) \models \spec$.
\end{definition}

\noindent
Here, the design is known to satisfy $\specv\subseteq \spec$. 
Furthermore, we are in good faith that $\design$ also satisfies $\spec 
\setminus \specv$, but it is not guaranteed.  We synthesize $\shield$, 
which reads the input and output of $\design$ while correcting its 
erroneous output as illustrated in Fig.~\ref{fig:attach_shield}.

\begin{definition}[Generic Shield]
Given a set $\spec = \specv \cup (\spec\setminus\specv)$ of properties.  
A reactive system $\shield$ is a \emph{generic shield} iff it is a 
shield of \emph{any} design $\design$ such that $\design \models 
\specv$.
\end{definition}

\noindent
A generic shield must work for any design $\design \models \specv$. 
Hence, the shield synthesis procedure does not need to consider the 
design implementation.  This is a realistic assumption in many 
applications, e.g., when the design $\design$ comes from the third 
party.  Synthesis of a generic shield also has a scalability advantage 
since the design $\design$, even if available, can be too complex to 
analyze, whereas $\spec$ often contains only a small set of critical
properties.  Finally, a generic shield is more robust against design 
changes, making it attractive for safety certification.  In this work, 
we focus on the synthesis of generic shields.

Although the shield is defined with respect to $\spec$ (more 
specifically, $\spec\setminus\specv$), we must refrain from ignoring the 
design completely while feeding the output with a replacement circuit.  
This is not desirable because the original design may satisfy additional 
(non-critical) properties that are not specified in $\spec$ but should 
be retained as much as possible.  In general, we want the shield to 
deviate from the design \emph{only if necessary, and as little as 
possible}.  For example, if $\design$ does not violate $\spec$, the 
shield $\shield$ should keep the output of $\design$ intact.  This 
rationale is captured by our next definitions.

\begin{definition}[Output Trace Distance Function]
An output trace distance function (OTDF) is a function $\distt: 
\doutalph^{*,\omega} \times \doutalph^{*,\omega} \rightarrow \NI$ such 
that
\begin{enumerate}
\item $\distt(\douttrace,\douttrace') = 0$ 
      when $\douttrace = \douttrace'$; 
\item $\distt(\douttrace\doutletter, \douttrace'\doutletter') = 
       \distt(\douttrace,\douttrace') $ 
      when $\doutletter = \doutletter'$, and
\item $\distt(\douttrace\doutletter, \douttrace'\doutletter') >
       \distt(\douttrace,\douttrace') $ 
      when $\doutletter \neq \doutletter'$.
\end{enumerate}
\end{definition}

\noindent
An OTDF measures the difference between two output sequences (of the 
design $\design$ and the shield $\shield$).  The definition requires 
monotonicity with respect to prefixes: when comparing trace prefixes 
with increasing length, the distance can only become larger.  

\begin{definition}[Language Distance Function]
A language distance function (LDF) is a function $\distl: \langset 
\times \dalph^\omega \rightarrow \NI$ such that 
$\forall \lang \in \langset, \dtrace \in \dalph^\omega \scope 
 \dtrace \in \lang \rightarrow \distl(\lang, \dtrace) = 0$.
\end{definition}

\noindent
An LDF measures the severity of specification violations by the design 
by mapping a language (of $\spec$) and a trace (of $\design$) to a 
number.  Given a trace $\dtrace\in\dalph^\omega$, its distance to 
$L(\spec)$ is 0 if $\dtrace$ satisfies $\spec$.  Greater distances 
indicate more severe specification violations.  An OTDF can (but does 
not have to) be defined via an LDF by taking the minimum output distance 
between $\dtrace = (\dintrace ||\douttrace)$ and any trace in the 
language $L$:
$$\distl(\lang, \dintrace ||\douttrace) =
\left\{
\begin{array}{ll}
\min\limits_{\dintrace || \douttrace' \in \lang} 
  \distt(\douttrace', \douttrace)
  &\quad\text{ if } \exists \douttrace'\in \doutalph^\omega \scope~
                (\dintrace||\douttrace') \in \lang \\
0
  &\quad\text{ otherwise.}
\end{array}
\right.
$$
The input trace is ignored in $\distt$ because the design $\design$ can 
only influence the output.  If no alternative output trace makes the 
word part of the language, the distance is set to $0$ to express that it 
cannot be the design's fault. If $\lang$ is defined by a realizable 
specification $\spec$, this cannot happen anyway, since $\forall 
\dintrace\in \dinalph^\omega \scope \exists \douttrace\in 
\doutalph^\omega \scope (\dintrace || \douttrace) \in \lang(\spec)$ is a 
necessary condition for the realizability of $\spec$.

\begin{definition}[Optimal Generic Shield]
\label{def:ogs}
Let $\spec$ be a specification, $\specv \subseteq \spec$ be the valid 
subset, $\distt$ be an OTDF, and $\distl$ be an LDF.  A reactive system 
$\shield$ is an \emph{optimal generic shield} if and only if for all 
$\dintrace \kin \dinalph^\omega$ and $\douttrace \kin 
\doutalph^\omega$, 
\begin{align}
(\dintrace || \douttrace) \kin \lang(\specv) \rightarrow & 
\bigl(
  \distl\bigl(~ \lang(\spec), 
              \dintrace || \shield(\dintrace || \douttrace) 
       ~\bigr) = 0
\;\wedge  \label{eq:kopt1} 
\\ &
 \distt(\douttrace, \shield(\dintrace || \douttrace)) \leq 
      \distl(\lang(\spec), \dintrace || \douttrace)
\bigr)
. \label{eq:kopt2}
\end{align}
\end{definition}

\noindent
The implication means that we only consider traces that satisfy $\specv$ 
since $\design \models \specv$ is assumed.   This can be exploited by 
synthesis algorithms to find a more succinct shield. 
Part~(\ref{eq:kopt1}) of the implied formula ensures correctness: 
$\design \comp \shield$ must satisfy $\spec$.\footnote{Applying $\distl$ 
instead of ``$\subseteq \lang(\spec)$'' adds flexibility: the user can 
define $\distl$ in such a way that $\distl(\lang, \dtrace) = 0$ even if 
$\dtrace\not\in\lang$ to allow such traces as well.} 
Part~(\ref{eq:kopt2}) ensures minimum interference: ``small'' violations 
result in ``small'' deviations. Def.~\ref{def:ogs} is designed to be 
flexible: Different notions of minimum interference can be realized with 
appropriate definitions of $\distt$ and $\distl$.  One realization will 
be presented in Section~\ref{sec:sol}.

\begin{proposition}
An optimal generic shield $\shield$ cannot deviate from the design's 
output before a specification violation by the design $\design$ is 
unavoidable.
\end{proposition}
\begin{proof}
If there has been a deviation $\distt(\douttrace, \shield(\dintrace || 
\douttrace)) \neq 0$ on the finite input prefix $\dtrace$, but this 
prefix can be extended into an infinite trace $\dtrace'$ such that 
$\distl(\lang(\spec), \dtrace') = 0$, meaning that a violation is 
avoidable, then Part~(\ref{eq:kopt2}) of Def.~\ref{def:ogs} is violated 
because of the (prefix-)monotonicity of $\distt$ (the deviation can only 
increase when the trace is extended), and the fact that $\distt \leq 
\distl$ is $\false$ if $\distt\neq 0$.
\end{proof}

\section{Our Shield Synthesis Method}
\label{sec:sol}

In this section, we present a concrete realization of the shield 
synthesis framework by defining OTDF and LDF in a practical way. We call 
the resulting shield a \emph{$k$-stabilizing} generic shield. While our 
framework works for arbitrary specifications, our realization assumes 
safety specifications.  

\subsection{$k$-Stabilizing Generic Shields}

A $k$-stabilizing generic shield is an optimal generic shield according 
to Def.~\ref{def:ogs}, together with the following restrictions.  When a 
property violation by the design $\design$ becomes unavoidable (in the 
worst case over future inputs), the shield $\shield$ is allowed to 
deviate from the design's outputs for at most $k$ consecutive time 
steps, including the current step.  Only after these $k$ steps, the next 
violation is tolerated. This is based on the assumption that 
specification violations are rare events.  If a second violation happens 
within the $k$-step recovery period, the shield enters a 
\emph{fail-safe} mode, where it enforces the critical properties, but 
stops minimizing the deviations.  
More formally, a $k$-stabilizing generic shield requires the following 
configuration of the OTDF and LDF functions:
\begin{enumerate}
\item The LDF $\distl(\lang(\spec),\dtrace) $ is defined as follows: 
  Given a trace $\dtrace \in \dalph^\omega$, its distance to
  $\lang(\spec)$ is $0$ initially, and increased to $\infty$ when the 
  shield enters the \emph{fail-safe} mode.
\item The OTDF function $\distt(\douttrace, \douttrace')$ returns $0$
  initially, and is set to $\infty$ if $\doutletter_i \neq 
  \doutletter_i'$ outside of a $k$-step recovery period.
\end{enumerate}
To indicate whether the shield is in the fail-safe mode or a recovery 
period, we add  a counter $c\in \{0,\ldots,k\}$. Initially, $c$ is 0. 
Whenever there is a property violation by the design, $c$ is set to $k$ 
in the next step.  In each of the subsequent steps, $c$ decrements until 
it reaches 0 again. The shield can deviate if the next state has $c>0$.  
If a second violation happens when $c>1$, then the shield enters the 
fail-safe mode.  A $1$-stabilizing shield can only deviate in the 
time step of the violation, and can never enter the fail-safe mode.

\subsection{Synthesizing $k$-Stabilizing Generic Shields}
\label{sec:k-stab}

The flow of our synthesis procedure is illustrated in 
Fig.~\ref{fig:k-stab}.  Let $\spec = \{\spec_1,\ldots,\spec_l\}$ be the 
critical safety specification, where each $\spec_i$ is represented as 
an automaton $\spec_i = (Q_i, q_{0,i}, \dalph, \delta_i,$ $F_i)$.  The 
synchronous product of these automata is again a safety automaton.  We 
use three product automata:
$\mathcal{Q} = (Q, q_{0}, \dalph, \delta, F)$ is the product of
all properties in $\spec$;
$\mathcal{V} = (V, v_{0}, \dalph, \delta^v, F^v)$ is the product
of properties in $\specv \subseteq \spec$; and 
$\mathcal{R} = (R, r_{0}, \dalph, \delta^r, F^r)$ is the product
of properties in $\spec\setminus\specv$.
Starting from these automata, our shield synthesis procedure consists of 
five steps.

\begin{figure}[tb]
  \begin{center}
    \includegraphics[width=0.85\textwidth]{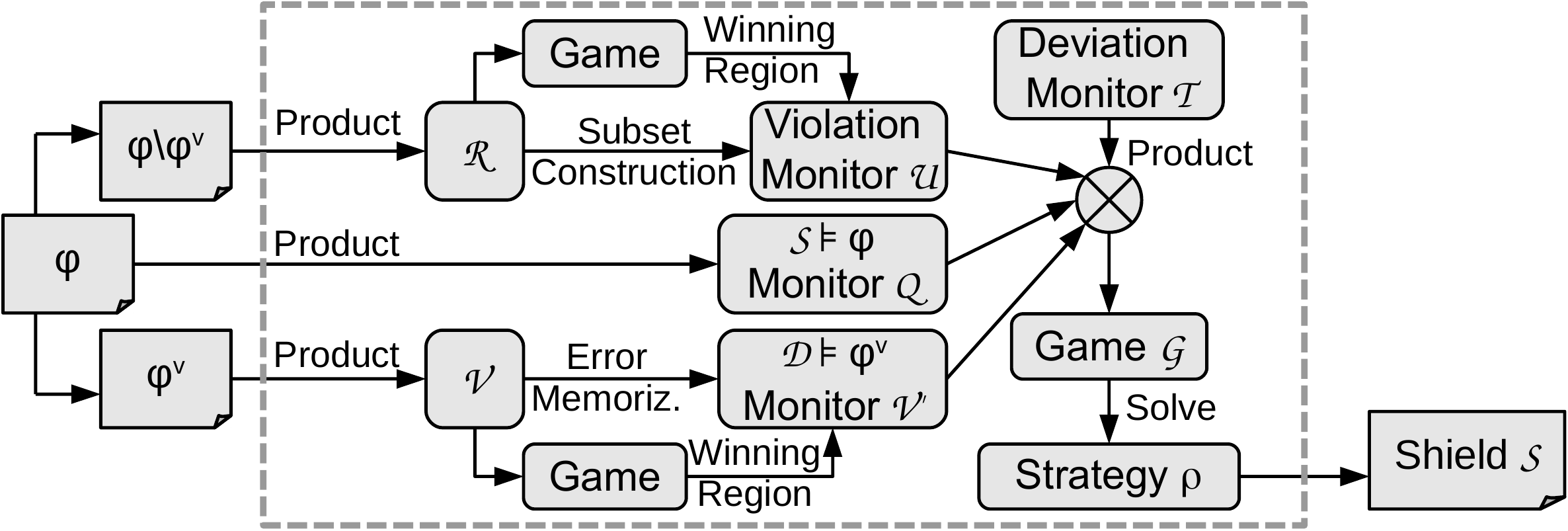}
    \caption{Outline of our $k$-stabilizing generic shield synthesis procedure.}
    \label{fig:k-stab}
  \end{center}
\end{figure}

\noindent
\textbf{Step 1. Constructing the Violation Monitor $\mathcal{U}$:}
From $\mathcal{R}$, which represents $\spec\setminus\specv$, we build 
$\mathcal{U} = (U, u_0, \dalph, \delta^u)$ to monitor property 
violations by the design.  The goal is to identify the latest point in 
time from which a specification violation can still be corrected with a 
deviation by the shield.  This constitutes the start of the recovery 
period.

The first phase of this construction (Step 1-a) is to consider the 
automaton $\mathcal{R} = (R, r_{0}, \dalph, \delta^r, F^r)$  as a 
\emph{safety game} and compute its winning region $W^r\subseteq F^r$.  
The meaning of $W^r$ is such that every reactive system 
$\design\models(\spec \setminus \specv)$ must produce outputs in such a 
way that the next state of $\mathcal{R}$ stays in $W^r$.  Only when the 
next state of $\mathcal{R}$ would be outside of $W^r$, our shield will 
be allowed to interfere.

\begin{exa}
Consider the safety automaton $\mathcal{R}$ in Fig.~\ref{fig:win_ex}, 
where $i$ is an input, $o$ is an output, and $r_x$ is unsafe.  The 
winning region is $W=\{r_0\}$ because from $r_1$ the input $i$ controls 
whether $r_x$ is visited.  The shield must be allowed to deviate from 
the original transition $r_0\rightarrow r_1$ if $o\neq i$. In $r_1$ it 
is too late because visiting an unsafe state cannot be avoided any more, 
given that the shield can modify the value of $o$ but not $i$.  \qed
\end{exa}

\begin{figure}[tb]
\centering
\begin{minipage}{0.38\textwidth}
\scalebox{0.9}{\input{figs/win_example.tex}}
\caption{The safety automaton $\mathcal{R}$.}
\label{fig:win_ex}
\end{minipage}
\hspace{0.08\textwidth}
\begin{minipage}{0.5\textwidth}
\scalebox{0.9}{\input{figs/dev_monitor.tex}}
\caption{The deviation monitor $\mathcal{T}$.}
\label{fig:dev_monitor}
\end{minipage}
\end{figure}

The second phase (Step 1-b) is to expand the state space from $R$ to 
$2^{R}$ via a subset construction.  The rationale behind it is as 
follows.  If the design makes a mistake (i.e., picks outputs such that 
$\mathcal{R}$ enters a state $r\not \in W^r$ from which the 
specification cannot be enforced), we have to ``guess'' what the design 
actually meant to do in order to find a state from which we can continue 
monitoring its behavior.  We follow a generous approach in order not to 
treat the design unfairly: we consider all output letters that would 
have avoided falling out of $W^r$, and continue monitoring the design 
behavior from all the corresponding successor states in parallel.  Thus, 
$\mathcal{U}$ is essentially a subset construction of $\mathcal{R}$, 
where a state $u\in U$ of $\mathcal{U}$ represents a set of states in 
$\mathcal{R}$.

The third phase (Step 1-c) is to expand the state space of $\mathcal{U}$ 
by adding a counter $c\in\{0,\ldots,k\}$ as described in the previous 
subsection, and adding a special fail-safe state $u_E$.   
The final violation monitor is $\mathcal{U} = (U, u_0, \dalph, 
\delta^u)$, where $U = (2^{R} \times \{0, \ldots, k\}) \cup {u_E}$ is 
the set of states, $u_0 = (\{r_0\}, 0)$ is the initial state, $\dalph$ 
is the set of input letters, and $\delta^u$ is the next-state function, 
which obeys the following rules:
\begin{enumerate}
\item $\delta^u(u_E, \dletter) = u_E$ 
      (meaning that $u_E$ is a trap state), \label{eq:subset_s}
\item $\delta^u((u,c), \dletter) = u_E$ if $c>1$ and
      $\forall r \in u: \delta^r(r,\dletter) \not\in W^r$, 
      \label{eq:subset_a}
\item $\delta^u((u,c), (\dinletter, \doutletter)) = 
      (\{r' \kin W^r \mid \exists r\in u, \doutletter' \in \doutalph 
         \scope \delta^r(r,(\dinletter,\doutletter')) = r'\}, k)$\\
if $c\leq 1$ and $\forall r \in u \scope 
    \delta^r(r,(\dinletter, \doutletter)) \not\in W^r$, 
and
\label{eq:subset_m}
\item $\delta^u((u,c), \dletter) \!=\! 
      (\{r'\kin W^r |\exists r\kin u \scope\delta^r(r,\dletter) = r'\}, 
       \textsf{dec}(c))$
if $\exists r \kin u \scope \delta^r(r,\dletter) \kin W^r$,
where $\textsf{dec}(0) = 0$ and $\textsf{dec}(c) = c-1$ if $c>0$.  
\label{eq:subset_n}
\end{enumerate}
Our construction sets $c=k$ whenever the design leaves the winning 
region, and not when it enters an unsafe state.  Hence, the shield 
$\shield$ can take remedial action as soon as the ``the crime is 
committed'', before the damage is detected, which would have been too 
late to correct the erroneous outputs of the design.

\begin{exa}
We illustrate the construction of $\mathcal{U}$ using the specification 
from Fig.~\ref{fig:traffic_spec}, 
\begin{wrapfigure}[10]{r}{0.42\textwidth}
\vspace{-0.7cm}
\begin{tabular}{l|cccccccc}
    &1g-   &1rg   &-rr &0gg     &0gr    &0rg    \\
\hline
H   &B\err &B\err &B   &HB\err  &H      &HB\err \\
B   &B\err &B\err &B   &HFB\err &H      &F      \\
F   &B\err &B\err &B   &FB\err  &FB\err &F      \\
HB  &B\err &B\err &B   &HFB\err &H      &F      \\
FB  &B\err &B\err &B   &HFB\err &H      &F      \\
HFB &B\err &B\err &B   &HFB\err &H      &F      \\
\hline
\end{tabular}
\caption{$\delta^u$ for the spec from Fig.~\ref{fig:traffic_spec}.}
\label{fig:traffic_subset}
\end{wrapfigure}
which is a safety automaton if we make all missing edges point to an 
(additional) unsafe state.  The winning region consists of all safe 
states, i.e., $W^r = \{H,B,F\}$.  The resulting violation monitor is 
$\mathcal{U}= 
(\{\text{H},\text{B},\text{F},\text{HB},\text{FB},\text{HFB}\} 
\times\{0, \ldots, k\}\cup u_E, (\text{H},0), \dalph, \delta^u)$, where 
$\delta^u$ is illustrated in Fig.~\ref{fig:traffic_subset} as a table
(the graph would be messy), which lists the next state for all possible 
present states as well as inputs and outputs by the design.  Lightning 
bolts denote specification 
violations. The update of the counter $c$, which is not included in 
Fig.~\ref{fig:traffic_subset}, is as follows: whenever the design 
commits a violation (indicated by lightning) and $c\leq 1$, then $c$ is 
set to $k$. If $c > 1$ at the violation, the next state is $u_E$.  
Otherwise, $c$ is decremented.  \qed

\end{exa}

\noindent
\textbf{Step 2. Constructing the Validity Monitor $\mathcal{V'}$:}
From $\mathcal{V}= (V, v_{0}, \dalph, \delta^v, F^v)$, which represents 
$\specv$, we build an automaton $\mathcal{V'}$ to monitor the validity 
of $\specv$ by solving a safety game on $\mathcal{V}$ and computing the 
winning region $W^v\subseteq F^v$. We will use $W^v$ to increase the 
freedom for the shield: since we assume that 
$\design\models\specv$, we are only interested in the cases where 
$\mathcal{V}$ never leaves $W^v$.  If it does, our shield is 
allowed to behave arbitrarily from that point on.  
We extend the state space from $V$ to $V'$ by adding a bit to 
memorize if we have left the winning region $W^v$.  Hence, the 
validity monitor is defined as $\mathcal{V}' = (V', v_{0}', \dalph, 
{\delta^v}', {F^v}')$, where $V' = \B \times V$ is the set of states, 
$v_{0}' = \{\false, v_0\}$ is the initial state, 
${\delta^v}'((b,v),\dletter) = (b', \delta^v(v,\dletter))$, where 
$b'=\true$ if $b=\true$ or $\delta^v(v,\dletter) \not \in W^v$, and 
$b'=\false$ otherwise, and  ${F^v}' = \{(b,v) \in V' \mid b=\false\}$.

\noindent
\textbf{Step 3. Constructing the Deviation Monitor $\mathcal{T}$:}
We build $\mathcal{T} = (T, t_0, \doutalph \times \doutalph, \delta^t)$ 
to monitor the deviation of the shield's output from the design's 
output. Here, $T = \{t_0, t_1\}$ and $\delta^t(t, (\doutletter, 
\doutletter')) = t_0$ iff $\doutletter = \doutletter'$. That is, 
$\mathcal{T}$ will be in $t_1$ if there was a deviation in the last time 
step, and in $t_0$ otherwise.  This deviation monitor is shown in 
Fig.~\ref{fig:dev_monitor}.

\noindent
\textbf{Step 4. Constructing the Safety Game $\mathcal{G}$:}
Given the monitors $\mathcal{U},\mathcal{V'},\mathcal{T}$ and the 
automaton $\mathcal{Q}$, which represents $\spec$, we construct a safety 
game $\mathcal{G} = (G, g_0, \dinalph \times \doutalph, \doutalph, $ 
$\delta^g, F^g)$, which is the synchronous product of $\mathcal{U}$, 
$\mathcal{T}$, $\mathcal{V}'$ and $\mathcal{Q}$, such that $G= U \times 
T \times V' \times Q$ is the state space, $g_0 = (u_0, t_0, v_0', q_0)$ 
is the initial state, $\dinalph\times\doutalph$ is the input of the 
shield, $\doutalph$ is the output of the shield, $\delta^g$ is the 
next-state function, and $F^g$ is the set of safe states, such that
$\delta^g\bigl((u, t, v', q), (\dinletter, \doutletter), 
\doutletter'\bigr) = $
\[
\bigl( \delta^u(u,(\dinletter, \doutletter)), 
       \delta^t(t,(\doutletter, \doutletter')), 
       {\delta^v}'(v',(\dinletter, \doutletter)), 
       \delta^q(q, (\dinletter, \doutletter'))
\bigr),
\]
and $F^g = \{(u, t, v', q)\in G \mid 
             v' \not\in {F^v}' \vee 
             ((q \in F^q) \wedge (u=(w,0) \rightarrow t=t_0))
           \}$.

In the definition of $F^g$, the term $v' \not\in {F^v}'$ reflects our 
assumption that $\design \models \specv$.  If this assumption is 
violated, then $v' \not\in {F^v}'$ will hold forever, and our shield is 
allowed to behave arbitrarily.  This is exploited by our synthesis 
algorithm to find a more succinct shield  by treating such states as 
\emph{don't cares}.  If $v' \in {F^v}'$, we require that $q \in F^q$, 
i.e., it is a safe state in $\mathcal{Q}$, which ensures that the shield 
output will satisfy $\spec$. The last term ensures that the shield can 
only deviate in the $k$-step recovery period, i.e., while $c\neq 0$ in 
$\mathcal{U}$.  If the design makes a second mistake within this period, 
$\mathcal{U}$ enters $u_E$ and arbitrary deviations are allowed.  Yet, 
the shield will still enforce $\spec$ in this mode (unless $\design 
\not\models \specv$).

\noindent
\textbf{Step 5. Solving the Safety Game:}  
We use standard algorithms for safety games (cf.  e.g.~\cite{Mazala01}) 
to compute a winning strategy $\rho$ for $\mathcal{G}$.  Then, we 
implement this strategy in a new reactive system $\shield = (G, g_0, 
\dalph, \doutalph, \delta, \rho)$ with $\delta(g, \dletter) = 
\delta^g(g, \dletter,\rho(g,\dletter))$. $\shield$ is the 
$k$-stabilizing generic shield.  If no winning strategy exists, we 
increase $k$ and try again.  In our experiments, we start with $k=1$ and 
then increase $k$ by 1 at a time.

\begin{theorem}
Let $\spec = \{\spec_1,\ldots,\spec_l\}$ be a set of critical safety 
properties $\spec_i = (\states_i, {\init}_i, $ $\dalph, \delta_i, 
F_i)$, and let $\specv \subseteq \spec$ be a subset of valid properties. 
 Let $|V| = \prod_{\spec_i \in \specv} |\states_i|$ be the cardinality 
of the product of the state spaces of all properties of $\specv$. 
Similarly, let $|R| = \prod_{\spec_i \not\in \specv} |\states_i|$. A 
$k$-stabilizing generic shield with respect to $\spec \setminus \specv$ 
and $\specv$ can be synthesized in $O(k^2 \cdot 2^{2|R|} \cdot |V|^4 
\cdot |R|^2)$ time
(if one exists).
\end{theorem}
\begin{proof}
Safety games can be solved in $O(x + y)$ time~\cite{Mazala01}, where $x$ 
is the number of states and $y$ is the number of edges in the game 
graph.  Our safety game $\mathcal{G}$ has at most $x = ((k+1) \cdot 
2^{|R|} + 1) \cdot (2 \cdot |V|) \cdot 2 \cdot (|R| \cdot |V|)$ states, 
so at most $y = x^2$ edges. 
\end{proof}

\noindent
\textbf{Variations.}
The assumption that no second violation occurs within the recovery 
period increases the chances that a $k$-stabilizing shield exists.  
However, it can also be dropped with a slight 
modification of $\mathcal{U}$ in Step 1: if a violation is committed and 
$c > 1$, we set $c$ to $k$ instead of visiting $u_E$.  This ensures 
that synthesized shields will handle violations within a 
recovery period normally.  
The assumption that the design meant to give one of the allowed outputs 
if a violation occurs can also be relaxed.  Instead of continuing to 
monitor the behavior from the allowed next states, we can just continue 
from the set of all states, i.e., traverse to state $(R,k)$ in 
$\mathcal{U}$.  The assumption that $\design\models\specv$, i.e., the 
design satisfies some properties, is also optional.  By removing 
$\mathcal{V}$ and $\mathcal{V}'$, the construction can be 
simplified at the cost of less implementation freedom for the shield.

By solving a \buchi game (which is potentially more expensive) instead 
of a safety game, we can also eliminate the need to increase $k$ 
iteratively until a solution is found. 
\ifextended
This is outlined in Appendix~\ref{sec:app}.
\else
This is outlined in the appendix of an extended version~\cite{extended} 
of this paper.
\fi

\section{Experiments} 
\label{sec:exp}

We have implemented the $k$-stabilizing shield synthesis procedure in a 
proof-of-concept tool.  Our tool takes as input a set of safety 
properties, defined as automata in a simple textual representation. The 
product of these automata, as well as the subset construction in Step 1 
of our procedure is done on an explicit representation.  The remaining 
steps are performed symbolically using Binary Decision Diagrams (BDDs).  
Synthesis starts with $k=1$ and increments $k$ in case of 
unrealizability until a user-defined bound is hit.  Our tool is written 
in Python and uses CUDD~\cite{Somenz95} as the BDD library.  Our tool 
can output shields in Verilog and SMV. It can also use the model checker 
VIS~\cite{VIS96} to verify that the synthesized shield is correct.

We have conducted three sets of experiments, where the benchmarks are 
(1) selected properties for a traffic light controller from the 
VIS~\cite{VIS96} manual, (2) selected properties for an ARM AMBA bus 
arbiter~\cite{BloemJPPS12}, and (3) selected properties from LTL 
specification patterns~\cite{DwyerAC99}.  None of these examples makes 
use of $\specv$, i.e., $\specv$ is always empty. The source code of our 
proof-of-concept synthesis tool as well as the input files and 
instructions to reproduce our experiments are available for 
download\footnote{\scriptsize\url{ 
http://www.iaik.tugraz.at/content/research/design_verification/others/}
}.

\noindent
\textbf{Traffic Light Controller Example.}
We used the safety specification in Fig.~\ref{fig:traffic_spec} as 
input, 
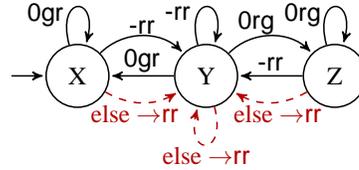
\begin{wrapfigure}[7]{r}{0.42\textwidth}
\vspace{-0.9cm}
\scalebox{0.95}{\input{figs/traffic_shield.tex}}
\vspace{-0.3cm}
\caption{Traffic light shield.}
\label{fig:traffic_shield}
\end{wrapfigure}
for which our tool generated a $1$-stabilizing shield within a fraction 
of a second.  The shield has 6 latches and 95 (2-input) multiplexers, 
which is then reduced by ABC~\cite{BraytonM10} to 5 latches and 41 
(2-input) AIG gates.  However, most of the states are either unreachable 
or equivalent.  The behavior of the shield is illustrated in 
Fig.~\ref{fig:traffic_shield}.  Edges are labeled with the inputs of the 
shield.  Red dashed edges denote situations where the output of the 
shield is different from its inputs.  The modified output is written 
after the arrow.  For all non-dashed edges, the input is just copied to 
the output.  Clearly, the states X, Y, and Z correspond to H, B, and F 
in Fig.~\ref{fig:traffic_spec}.

We also tested the synthesized shield using the traffic light controller 
of~\cite{vlsi}, which also appeared in the user manual of 
VIS~\cite{VIS96}. This controller has one input (\textsf{car}) from a 
car sensor on the farm road, and uses a timer to control the length of 
the different phases.  We set the ``short'' timer period to one tick and 
the ``long'' period to two ticks.  

\begin{wrapfigure}[8]{r}{0.58\textwidth}
\vspace{-0.9cm}
\scalebox{0.8}{\input{figs/traffic_impl.tex}}
\caption{Traffic light implementation.}
\label{fig:traffic_impl}
\end{wrapfigure}
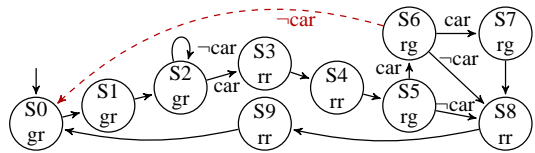

The resulting behavior without preemption is visualized in 
Fig.~\ref{fig:traffic_impl}, where nodes are labeled with names and 
outputs, and edges are labeled with conditions on the inputs.  The red 
dashed arrow represents a subtle bug we introduced: if the last car on 
the farm road exits the crossing at a rare point in time, then the 
controller switches from \textsf{rg} to \textsf{gr} without passing 
\textsf{rr}. This bug only shows up in very special situations, so it 
can go unnoticed easily.  Preemption is implemented by modifying both 
directions to \textsf{r} without changing the state if $\textsf{p}=1$.  
We introduced another bug here as well: only the highway is switched to 
\textsf{r} if $\textsf{p}=1$, whereas the farm road is not. This bug can 
easily go unnoticed as well, because the farm road is mostly red anyway.
The following trace illustrates how the synthesized shield handles these 
errors:

\begin{center}
\setlength{\tabcolsep}{1.3pt}
\begin{tabular}{l|c|c|c|c|c|c|c|c|c|c|c|c|c|c|c|c}\hline
Step                                 &0 &1 &2 &3 &4 &5 &6 &7     &8  &9 &10&11&12    &13&14&15    \\
\hline
State in Fig.~\ref{fig:traffic_spec} (safety spec.)
                                     &H &H &B &H &B &B &F &F     &F,B&H &H &B &B     &B &B &\ldots\\
State in Fig.~\ref{fig:traffic_impl} (buggy design)
                                     &S0&S1&S2&S3&S4&S5&S6&S0    &S1 &S2&S3&S4&S5    &S8&S9&\ldots\\
State in Fig.~\ref{fig:traffic_shield} (shield)
                                     &X &X &Y &X &Y &Y &Z &Z     &Y  &X &X &Y &Y     &Y &Y &\ldots\\

Input (\textsf{p,car})               &00&11&01&01&01&01&00&00    &00 &01&01&00&10    &00&00&\ldots\\
Design output                        &gr&rr&gr&rr&rr&rg&rg&gr\err&gr &gr&rr&rr&rg\err&rr&rr&\ldots\\
Shield output                        &gr&rr&gr&rr&rr&rg&rg&rr    &gr &gr&rr&rr&rr    &rr&rr&\ldots\\
\hline
\end{tabular}
\end{center}

\noindent
The first bug strikes at Step 7. The shield corrects it with output 
\textsf{rr}. A $2$-stabilizing shield could also have chosen 
\textsf{rg}, but this would have made a second deviation necessary in 
the next step. Our shield is $1$-stabilizing, i.e., it deviates only at 
the step of the violation. After this correction, the shield continues 
monitoring the design from both state F and state B of 
Fig.~\ref{fig:traffic_spec}, as explained earlier, to detect future 
errors. Yet, this uncertainty is resolved in the next step.  The second 
bug in Step 12 is simpler: outputting \textsf{rr} is the only way to 
correct it, and the next state in Fig.~\ref{fig:traffic_spec} must be B.

When only considering the properties 1 and 2 from Section~\ref{sec:ex}, 
the synthesized shield has no latches and three AIG gates after 
optimization with ABC~\cite{BraytonM10}.

\noindent
\textbf{ARM AMBA Bus Arbiter Example.}
We used properties of an ARM AMBA bus arbiter~\cite{BloemJPPS12} as 
input to our shield synthesis tool.  Due to page limit, we only present 
the result on one example property, and then present the performance 
results for other properties. 
The property that we enforced was Guarantee 3 from the specification 
of~\cite{BloemJPPS12}, which says that if a length-four locked burst 
access starts, no other access can start until the end of this burst. 
The safety automaton is shown in Fig.~\ref{fig:amba_g3}, where 
\texttt{B}, \texttt{s} and \texttt{R} are short for $\texttt{hmastlock} 
\wedge \texttt{HBURST=BURST4}$, \texttt{start}, and \texttt{HREADY}, 
respectively.  Lower case signal names are outputs, and upper-cases are 
inputs of the arbiter.  S$_x$ is unsafe. S0 is the idle state waiting 
for a burst to start ($\texttt{B} \wedge \texttt{s}$). The burst is over 
if input \texttt{R} has been $\true$ $4$ times.  State S$i$, where 
$i=1,2,3,4$, means that \texttt{R} must be $\true$ for $i$ more times.  
The counting includes the time step where the burst starts, i.e., where 
S0 is left.  Outside of S0, \texttt{s} is required to be $\false$.

\begin{figure}[tb]
\centering
\begin{minipage}{0.47\linewidth}
\vspace{-0.3cm}
\hspace{-1ex}
\scalebox{0.8}{\input{figs/amba_g3.tex}}
\vspace{-0.4cm}
\caption{Guarantee 3 from~\cite{BloemJPPS12}.}
\label{fig:amba_g3}
\end{minipage}
\begin{minipage}{0.507\linewidth}
\vspace{-2mm}
\setlength{\tabcolsep}{1.1pt}
\begin{tabular}{l|c|c|c|c|c|c|c|c|c|c}\hline
Step                           &3 &4 &5 &6 &7    &8 &9 &10&11&12    \\
\hline
State in Fig.~\ref{fig:amba_g3}&S0&S4&S3&S2&S1   &S0&S0&S0&S0&\ldots\\
State in Design                &S0&S3&S2&S1&S0   &S3&S2&S1&S0&\ldots\\
\texttt{B}                     &1 &1 &1 &1 &1    &1 &1 &1 &1 &\ldots\\
\texttt{R}                     &0 &1 &1 &1 &1    &1 &1 &1 &1 &\ldots\\
\texttt{s} from Design         &1 &0 &0 &0 &1\err&0 &0 &0 &0 &\ldots\\
\texttt{s} from Shield         &1 &0 &0 &0 &0    &0 &0 &0 &0 &\ldots\\
\hline
\end{tabular}
\caption{Shield execution results.}
\label{fig:amba_g3_result}
\end{minipage}
\end{figure}

Our tool generated a 1-stabilizing shield within a fraction of a second. 
The shield has 8 latches and 142 (2-input) multiplexers, which is then 
reduced by ABC~\cite{BraytonM10} to 4 latches and 77 AIG gates. We 
verified it against an arbiter implementation for 2 bus masters, where 
we introduced the following bug: the design does not check \texttt{R} 
when the burst starts, but behaves as if \texttt{R} was $\true$.  This 
corresponds to removing the transition from S0 to S4 in 
Fig.~\ref{fig:amba_g3}, and going to S3 instead.  An execution trace is 
shown in Fig.~\ref{fig:amba_g3_result}.  The first burst starts with 
$\texttt{s}=\true$ in Step 3.  \texttt{R} is $\false$, so the design 
counts wrongly.  The erroneous output shows up in Step 7, where the 
design starts the next burst, which is forbidden, and thus blocked by 
the shield.  The design now thinks that it has started a burst, so it 
keeps $\texttt{s}=\false$ until \texttt{R} is $\true$ 4 times.  
Actually, this burst start has been blocked by the shield, so the shield 
waits in S0.  Only after the suppressed burst is over, the components 
are in sync again, and the next burst can start normally.

\begin{wraptable}[10]{r}{0.46\textwidth}
\centering
\setlength{\tabcolsep}{5pt}
\vspace{-5ex}
\caption{Performance for AMBA~\cite{BloemJPPS12}.}
\vspace{-0.3cm}
\label{tab:perf}
\scalebox{0.8}{\begin{tabular}{|lccccc|}
\hline
Property        &$|Q|$ &$|I|$    &$|O|$      &$k$ &Time [sec] \\
\hline
G1              &3     &1        &1         &1   &0.1        \\
G1+2            &5     &3        &3         &1   &0.1        \\
G1+2+3          &12    &3        &3         &1   &0.1        \\
G1+2+4          &8     &3        &6         &2   &7.8        \\
G1+3+4          &15    &3        &5         &2   &65         \\
G2+3+4          &17    &3        &6         &?   &$>$3600    \\
G1+2+3+5        &18    &3        &4         &2   &242        \\
G1+2+4+5        &12    &3        &7         &?   &$>$3600    \\
G1+3+4+5        &23    &3        &6         &?   &$>$3600    \\
\hline
\end{tabular}}
\end{wraptable}
To evaluate the performance of our tool, we ran a stress test with 
increasingly larger sets of safety properties for the ARM AMBA bus 
arbiter in~\cite{BloemJPPS12}.  Table~\ref{tab:perf} summarizes the 
results.  The columns list the number of states, inputs, and outputs, 
the minimum $k$ for which a $k$-stabilizing shield exists, and the 
synthesis time in seconds.  All experiments were performed on a machine 
with an Intel i5-3320M CPU@2.6 GHz, 8 GB RAM, and a 64-bit Linux.  
Time-outs (G2+3+4, G1+2+4+5 and G1+3+4+5) occurred only when the number 
of states and input/output signals grew large.  However, this should not 
be a concern in practice because the set of critical properties of a 
system is usually much smaller, e.g., often consisting of invariance 
properties with a single state.  

\begin{wraptable}[18]{r}{0.631\textwidth}
\centering
\setlength{\tabcolsep}{1pt}
\vspace{-5ex}
\caption{Synthesis results for the LTL patterns~\cite{DwyerAC99}.}
\label{tab:dwyer}
\vspace{-2ex}
\scalebox{0.85}{
\begin{tabular}{|l|l|c||c|c|c|c|}
\hline
Nr. &Property   &$b$&$|Q|$ &Time      &\#Lat-    &\#AIG-      \\
    &           &   &      &[sec]     &ches      &Gates       \\
\hline
1   &$\always \neg p$
                &-  &2     &0.01      &0         &0           \\
2   &$\eventually r \rightarrow (\neg p \until r)$
                &-  &4     &0.34      &2         &6           \\
3   &$\always(q \rightarrow \always(\neg p))$
                &-  &3     &0.34      &2         &6           \\
4   &$\always((q \wedge \neg r \wedge \eventually r) \rightarrow 
              (\neg p \until r))$
                &-  &4     &0.34      &1         &9           \\
5   &$\always(q \wedge \neg r \rightarrow (\neg p \weakuntil r))$
                &-  &3     &0.01      &2         &14          \\
6   &$\eventually p$
                &0  &3     &0.34      &1         &1           \\
6   &$\eventually p$
                &256&259   &33        &18        &134         \\
7   &$\neg r \weakuntil (p \wedge \neg r)$
                &-  &3     &0.05      &3         &11          \\
8   &$\always(\neg q) \vee \eventually(q \wedge \eventually p)$
                &0  &3     &0.04      &3         &11          \\
8   &$\always(\neg q) \vee \eventually(q \wedge \eventually p)$
                &4  &7     &0.04      &6         &79          \\
8   &$\always(\neg q) \vee \eventually(q \wedge \eventually p)$
                &16 &19    &0.03      &10        &162         \\
8   &$\always(\neg q) \vee \eventually(q \wedge \eventually p)$
                &64 &67    &0.37      &14        &349         \\
8   &$\always(\neg q) \vee \eventually(q \wedge \eventually p)$
                &256&259   &34        &18        &890         \\
9   &$\always(q \wedge \neg r \rightarrow 
              (\neg r \weakuntil (p \wedge \neg r))) $
                &-  &3     &0.05      &2         &12          \\
10   &$\always(q \wedge \neg r \rightarrow 
              (\neg r \until (p \wedge \neg r))) $
                &12 &14    &5.4       &14        &2901        \\
10   &$\always(q \wedge \neg r \rightarrow 
              (\neg r \until (p \wedge \neg r))) $
                &14 &16    &38        &15        &6020        \\
10   &$\always(q \wedge \neg r \rightarrow 
              (\neg r \until (p \wedge \neg r))) $
                &16 &18    &377       &18        &13140       \\
\hline                
\end{tabular}
}
\end{wraptable}

\noindent
\textbf{LTL Specification Patterns.}
Dwyer et al.~\cite{DwyerAC99} studied the frequently used LTL 
specification patterns in verification.  As an exercise, we applied our 
tool to the first 10 properties from their list~\cite{LTLpattern-URL} 
and summarized the results in Table~\ref{tab:dwyer}.  For a property 
containing liveness aspects (e.g., something must happen eventually), we 
imposed a bound on the reaction time to obtain the safety 
(bounded-liveness) property. The bound on the  reaction time is shown in 
Column~3.  The last four columns list the number of states in the safety 
specification, the synthesis time in seconds, and the shield size 
(latches and AIG gates). Overall, our method runs sufficiently fast on 
all properties and the resulting shield size is small. We also 
investigated how the synthesis time increased with an increasingly 
larger bound $b$.  For Property 8 and Property 6, the run time and 
shield size remained small even for large automata. For Property 10, the 
run time and shield size grew faster, indicating room for further 
improvement.  As a proof-of-concept implementation, our tool has not yet 
been optimized specifically for speed or shield size -- we leave such 
optimizations for future work.

\section{Conclusions}
\label{sec:conc}

We have formally defined the shield synthesis problem for reactive 
systems and presented a general framework for solving the problem.  We 
have also implemented a new synthesis procedure that solves  a concrete 
instance of  this problem, namely the synthesis of $k$-stabilizing 
generic shields. We have evaluated our new method  on two hardware 
benchmarks and a set of LTL specification patterns.  We believe that our 
work points to an exciting new direction for applying synthesis, because 
the set of critical properties of a complex system tends to be small and 
relatively easy to specify, thereby making shield synthesis scalable and 
usable. Many interesting extensions and variants remain to be explored, 
both theoretically and experimentally, in the future.

\bibliography{references}
\ifextended
\newpage
\appendix

\section{Synthesis of Stabilizing Generic Shields}
\label{sec:app}

In this section, we present a method for synthesizing $k$-stabilizing 
shields with arbitrary but finite $k$.  We call such shields 
\emph{stabilizing} (without the ``$k$'').  A synthesis procedure for 
stabilizing shields is also useful as a preprocessing step if we want to 
enforce a particular (or minimal) $k$: Even for a realizable 
specification, the $k$-stabilizing shield synthesis problem may be 
unrealizable for any finite $k$.  When specification $\spec$ is 
realizable, there exists a reactive system $\design'$ such that 
$\design' \models \spec$.  However, it does not mean that a shield 
$\shield$ exists for any design $\design$, such that $(\design \comp 
\shield) \models \spec$, and $(\design \comp \shield)$ deviates from 
$\design$ for at most $k$ time steps.

\begin{exa}
Consider the safety specification on the right, where $o_1$ and $o_2$ 
are 
\begin{wrapfigure}[6]{r}{0.42\textwidth}
\vspace{-1.0cm}
\input{figs/unreal_example.tex}
\end{wrapfigure}
outputs, and $r_x$ is unsafe.  The design must produce either $o_1 
\wedge \neg o_2$ globally or $\neg o_1$ globally.  The $k$-stabilizing 
shield synthesis problem is unrealizable for any finite $k$: if the 
design produces $o_1 \wedge o_2$ initially, the shield must deviate to 
either $o_1 \wedge \neg o_2$ or $\neg o_2$.  In the former case, the 
design could produce $\neg o_1$ from that point on, in the latter case 
$o_1 \wedge \neg o_2$.  This would cause an indefinite deviation with 
only a single violation.  \qed

\end{exa}

Whether a $k$-stabilizing shield exists for some finite $k$ is difficult 
to detect with the synthesis procedure from Section~\ref{sec:k-stab}.  
In case of unrealizability of the shield for a given $k$, we cannot know 
if we just need to increase $k$, or if no finite $k$ would work.  The 
synthesis process presented in the following sub-section will decide the 
realizability problem.  We can also synthesize a stabilizing shield, 
measure its $k$, and minimize this $k$ further with the procedure from 
Section~\ref{sec:k-stab} until we hit the unrealizability barrier.

\subsection{Construction for Synthesizing Stabilizing Shields} 
\label{sec:stab-app}

A generic stabilizing shield can be synthesized (if one exists) with 
only a few modifications to the procedure from Section~\ref{sec:k-stab}. 
Instead of a counter $c\in\{0,\ldots, k\}$, we use a counter 
$d\in\{0,1,2\}$ with only three different values.  Intuitively, $d=2$ is 
an abstraction for $c > 1$.  
We construct a \buchi game that is won if $d\leq 1$ infinitely often 
(and all the other shield requirements are satisfied).  A \buchi game is 
like a safety game, but the given set of final states must be visited 
infinitely often for the system to win the game.  A winning strategy for 
this \buchi game corresponds to a $k$-stabilizing shield with some 
finite $k$, and the $k$ can even be computed during synthesis.  The 
construction is similar to Section~\ref{sec:k-stab}, with only a few 
modifications:

\noindent
\textbf{Step 1.} Instead of using a counter $c\in\{0,\ldots, k\}$, we 
use a three-valued counter $d\in\{0,1,2\}$ to track whether we are 
currently in the recovery phase or not. Intuitively, $d=2$ if $c$ would 
be $>1$. That is, $d$ is $0$ initially.  If $d < 2$ and the design 
makes a mistake (leaves $W^r$), then $d$ is set to $2$.  If it was 
already $2$, we enter $u_E$.  In order to decide when to decrement $d$ 
from $2$ to $1$, we add a special output $r$ to the shield.  If this 
output is set to $\true$ and $d=2$, then $d$ is set to $1$ in the next 
step. The behavior for $d=1$ is the same as in Section~\ref{sec:k-stab}: 
if another violation occurs, $d$ is set to $2$. Otherwise, $d$ is 
decremented to $1$.  We denote this slightly modified violation monitor 
by $\mathcal{U}' = (U', u_0', \dalph \times 2^{\{r\}}, {\delta^u}')$ 
with $U'=(2^{R} \times \{0,1,2\}) \cup {u_E}$. The subsequent steps will 
ensure that the shield will only be allowed to deviate if $d > 0$ in the 
next step.  We will also require that $d$ cannot be $2$ indefinitely.

\noindent
\textbf{Step 2 and Step 3} are performed as described in 
Section~\ref{sec:k-stab}.

\noindent
\textbf{Step 4.} We construct a \buchi game $\mathcal{G}' = (G', g_0', 
\dinalph \times \doutalph, \doutalph \times 2^{\{r\}}, {\delta^g}', 
{F^g}')$ as the synchronous product of $\mathcal{U}'$, $\mathcal{T}$, 
$\mathcal{V}'$ and $\mathcal{Q}$ as follows: 
\begin{itemize}
\item $G'= U' \times T \times V'  \times Q \times \B \times \B$, 
\item $g_0' = (u_0', t_0, v_0', q_0, \false, \false)$,
\item ${\delta^g}'\bigl((u', t, v', q, m, n), (\dinletter, \doutletter), 
(\doutletter',r) \bigr) =
\bigl(
{\delta^u}'(u',((\dinletter, \doutletter),r)),\\
\delta^t(t,(\doutletter, \doutletter')),
{\delta^v}'(v',(\dinletter, \doutletter)),
\delta^q(q, (\dinletter, \doutletter')),m',n'
\bigr)$,
where
 \begin{itemize}
 \item $m'=\true$ iff $m=\true$ or $q\not \in F^q$
 \item $n'=\true$ iff $n=\true$ or $u'=(w,0) \wedge t=t_1$, and
 \end{itemize}
\item ${F^g}' = \{(u',t,v',q,m,n)\in G' \mid v'\not\in {F^v}' \vee
                   (\neg n \wedge \neg m \wedge d \leq 1)\}$.
\end{itemize}
The intuition behind this construction is as follows.  We extend the 
state space of the synchronous product by two bits, $m$ and $n$.  The 
bit $m$ is $\true$ if the execution has ever visited an unsafe state in 
$\mathcal{Q}$.  The bit $n$ is $\true$ if there has been an illegal 
deviation\footnote{Recall that $u'=(w,0)$ means that the counter $d$ 
introduced in Step 1 is $0$, i.e., no deviation was allowed in the 
previous time step; $t=t_1$ indicates that a deviation has occurred in 
the previous time step.}. With this information, the accepting states 
(that need to be visited infinitely often) are then defined as follows. 
Outside of ${F^v}'$, everything is accepting.  This makes sure that the 
shield can behave arbitrarily if $\design\not\models\specv$.  Otherwise, 
a state is accepting if $d\leq 1$ (the last recovery period is over), 
$m$ is $\false$ ($\design\comp\shield \models \spec$ so far) and $n$ is 
$\false$ (no illegal deviations so far).  Visiting ${F^g}'$ infinitely 
often implies that recovery periods are over infinitely often, and $m$ 
and $n$ are never $\true$ (these bits cannot change back to $\false$).

\noindent
\textbf{Step 5.}  Just like safety games, \buchi games also have a 
memoryless strategy.  We compute such a strategy and implement it as 
described in Section~\ref{sec:k-stab}. If no such strategy exists (which 
is easy to detect during synthesis), then this is reported to the user.

\noindent
\textbf{Discussion.}
Note that the \buchi objective ensures that recovery phases are over 
infinitely often, but not that they are bounded in time.  There may 
exist a strategy to satisfy the \buchi objective without any finite 
bound on the recovery time.  E.g., the first recovery phase could take 2 
steps, the second one 4 steps, the third one 8 steps, etc.  However, 
such a strategy would require infinite memory. We construct and 
implement a memoryless strategy, which guarantees a bounded recovery.  
We can even measure the maximum length of any recovery phase while 
synthesizing the shield: \buchi games can be solved with a doubly-nested 
fixpoint computation~\cite{Mazala01}.  The number of iterations of the 
inner fixpoint (in the last iteration of the outer fixpoint) corresponds 
to the maximum number of steps needed to reach a state of ${F^g}'$, 
i.e., a state where the recovery is over.  Hence, this value is also the 
maximum length of a recovery period, i.e., the value $k$ for the 
resulting $k$-stabilizing shield.
\fi
\end{document}

%% file: preamble.tex
\usepackage{times}
\usepackage{amsthm}
\usepackage{booktabs}
\usepackage{rotating}
\usepackage{xcolor}
\definecolor{darkred}{rgb}{0.7, 0.0, 0.0}
\usepackage{caption}
\usepackage{listings}
\usepackage{xspace}
\usepackage{hyperref}
\usepackage{amsfonts}
\usepackage{amsmath}
\newcommand{\comment}[1]{}

\usepackage{amssymb}
\usepackage{temporal}
\usepackage{hyperref}
\usepackage{breakurl}
\usepackage{stmaryrd}
\usepackage{enumerate}
\usepackage{wrapfig} 

\newcounter{exacounter}
\newenvironment{exa}{
\refstepcounter{exacounter}
\smallskip\noindent
\textbf{Example \theexacounter.}
}{\vspace{2mm}}

\usepackage{tikz}
\usetikzlibrary{arrows,automata}
\newcommand{\nd}{3cm} 
\tikzset{initial text={}}
\tikzset{every picture/.style=semithick} 
\tikzset{>=stealth'} 
\tikzset{->} 
\tikzset{shorten >=1pt} 

\newcommand{\buchi}{B\"uchi\xspace}
\newcommand{\win}{\mathsf{win}}
\newcommand{\B}{\mathbb{B}}
\newcommand{\N}{\mathbb{N}}
\newcommand{\NI}{\mathbb{N}^{\infty}}

\newcommand{\design}{\mathcal{D}}
\newcommand{\shield}{\mathcal{S}}
\newcommand{\game}{\mathcal{G}}

\newcommand{\gstates}{G}
\newcommand{\ginit}{g_0}
\newcommand{\states}{Q}
\newcommand{\init}{q_0}
\newcommand{\din}{I}
\newcommand{\dinalph}{\Sigma_I}
\newcommand{\dinletter}{{\sigma_I}}
\newcommand{\dintrace}{{\overline{\sigma_I}}}
\newcommand{\dout}{O}

\newcommand{\doutalph}{\Sigma_O}
\newcommand{\doutletter}{{\sigma_O}}
\newcommand{\douttrace}{{\overline{\sigma_O}}}
\newcommand{\dalph}{\Sigma}
\newcommand{\dletter}{\sigma}
\newcommand{\dtrace}{\overline{\dletter}}
\newcommand{\lang}{L}
\newcommand{\langset}{\mathcal{L}}
\newcommand{\spec}{\varphi}
\newcommand{\specv}{\varphi^v}
\newcommand{\distt}{d^\sigma}
\newcommand{\distl}{d^L}

\newcommand{\kin}{\!\in\!}
\newcommand{\err}{{\color{darkred}$\lightning$}}
\newcommand{\comp}{\circ}

%% file: shield_spec2.pstex_t
\begin{picture}(0,0)%
\includegraphics{./shield_spec2.pstex}%
\end{picture}%
\setlength{\unitlength}{3947sp}%
\begingroup\makeatletter\ifx\SetFigFont\undefined%
\gdef\SetFigFont#1#2#3#4#5{%
  \reset@font\fontsize{#1}{#2pt}%
  \fontfamily{#3}\fontseries{#4}\fontshape{#5}%
  \selectfont}%
\fi\endgroup%
\begin{picture}(2324,1565)(1179,-1122)
\put(1201,-811){\makebox(0,0)[lb]{\smash{{\SetFigFont{14}{16.8}{\familydefault}{\mddefault}{\updefault}{\color[rgb]{0,0,0}$\textsf{f}$}%
}}}}
\put(1201,-300){\makebox(0,0)[lb]{\smash{{\SetFigFont{14}{16.8}{\familydefault}{\mddefault}{\updefault}{\color[rgb]{0,0,0}$\textsf{h}$}%
}}}}
\put(1201,239){\makebox(0,0)[lb]{\smash{{\SetFigFont{14}{16.8}{\familydefault}{\mddefault}{\updefault}{\color[rgb]{0,0,0}$\textsf{p}$}%
}}}}
\put(3241,164){\makebox(0,0)[lb]{\smash{{\SetFigFont{14}{16.8}{\familydefault}{\mddefault}{\updefault}{\color[rgb]{0,0,0}$\textsf{h'}$}%
}}}}
\put(3241,-511){\makebox(0,0)[lb]{\smash{{\SetFigFont{14}{16.8}{\familydefault}{\mddefault}{\updefault}{\color[rgb]{0,0,0}$\textsf{f'}$}%
}}}}
\end{picture}%

%% file: figs/traffic_spec.tex
\begin{tikzpicture}[auto,node distance=\nd]
\node[state,initial]  at  (0,0)       (H) {H};
\node[state]          at  (1.8,0)       (B) {B};
\node[state]          at  (3.6,0)       (F) {F};

\path
(H) edge [loop above]  node[xshift=-4.5mm,yshift=-4mm] {\textsf{0gr}} (H)
(B) edge [loop above]  node[xshift=-4mm,yshift=-3mm] {\textsf{-rr}} (B)
(F) edge [loop above]  node[xshift=-4.5mm,yshift=-4mm] {\textsf{0rg}} (F)
(H) edge [bend left]   node[xshift=0mm,yshift=-1mm] {\textsf{-rr}} (B)
(B) edge [bend left]   node[xshift=0mm,yshift=5mm] {\textsf{0gr}} (H)
(F) edge [bend left]   node[xshift=0mm,yshift=4mm] {\textsf{-rr}} (B)
(B) edge [bend left]   node[xshift=-1mm,yshift=-1mm] {\textsf{0rg}} (F);
\end{tikzpicture}

%% file: figs/win_example.tex
\begin{tikzpicture}[auto,node distance=\nd]
\node[state,initial] at  (0,0)       (r0) {$r_0$};
\node[state]         at  (2,0)       (r1) {$r_1$};
\node[state]         at  (3.4,0)     (rx) {$r_x$};

\path
(r0) edge [loop above]  
     node [xshift=-6mm,yshift=-4mm] {$o=i$} (r0)
(r0) edge               
     node [xshift=0mm,yshift=-5mm]  {$o\neq i$} (r1)
(r1) edge [bend right]
     node [xshift=0mm,yshift=4mm]   {$\neg i$} (r0)
(r1) edge               
     node                           {$i$} (rx)
(rx) edge [loop above]  
     node [xshift=-5mm,yshift=-4mm] {$\true$} (rx)
     
;
\end{tikzpicture}

%% file: figs/dev_monitor.tex
\begin{tikzpicture}[auto,node distance=\nd]
\node[initial,state] at  (0,0)       (t0) {$t_0$};
\node[state]         at  (3,0)       (t1) {$t_1$};

\path 
(t0) edge [loop above] 
     node [xshift=-9mm,yshift=-4mm] {$\doutletter  =  \doutletter'$} (t0)
(t0) edge 
     node [xshift=0mm,yshift=-5mm] {$\doutletter \neq \doutletter'$} (t1)
(t1) edge [loop above]  
     node [xshift=10mm,yshift=-4mm] {$\doutletter \neq \doutletter'$} (t1)
(t1) edge [bend right] 
     node [xshift=1mm,yshift=5mm] {$\doutletter  =  \doutletter'$} (t0)
;
\end{tikzpicture}

%% file: figs/traffic_shield.tex
\begin{tikzpicture}[auto,node distance=\nd]
\definecolor{darkred}{rgb}{0.7, 0.0, 0.0}
\node[state,initial]  at  (0,0)         (X) {X};
\node[state]          at  (1.8,0)       (Y) {Y};
\node[state]          at  (3.6,0)       (Z) {Z};

\path
(X) edge [loop above]
    node[xshift=-4.5mm,yshift=-4mm] {\textsf{0gr}} (X)
(Y) edge [loop above]                 
    node[xshift=-4mm,yshift=-3mm] {\textsf{-rr}} (Y)
(Y) edge [loop below,dashed,color=darkred]                 
    node[xshift=0mm,yshift=1mm] {else $\rightarrow$\textsf{rr}} (Y)    
(Z) edge [loop above]                 
    node[xshift=-4.5mm,yshift=-4mm] {\textsf{0rg}} (Z)
(X) edge [bend angle=45, bend left]   
    node[xshift=0mm,yshift=-1mm] {\textsf{-rr}} (Y)
(X) edge [bend angle=30, bend right, dashed,color=darkred]       
    node[xshift=-1mm,yshift=-4mm] {else $\rightarrow$\textsf{rr}} (Y)
(Y) edge                              
    node[xshift=0mm,yshift=5mm] {\textsf{0gr}} (X)
(Z) edge                   
    node[xshift=0mm,yshift=4mm] {\textsf{-rr}} (Y)
(Z) edge [bend angle=30, bend left, dashed,color=darkred]       
    node[xshift=1mm,yshift=1mm] {else $\rightarrow$\textsf{rr}} (Y)
(Y) edge [bend angle=45, bend left]                  
    node[xshift=-2mm,yshift=-1mm] {\textsf{0rg}} (Z);
\end{tikzpicture}

%% file: figs/traffic_impl.tex
\begin{tikzpicture}[auto,node distance=\nd]
\definecolor{darkred}{rgb}{0.7, 0.0, 0.0}
\node[state,initial above,inner sep=0pt,align=center] at  (0,0)  (S0)  {S0\\gr};
\node[state,inner sep=0pt,align=center]          at  (1.2,0.3)   (S1)  {S1\\gr};
\node[state,inner sep=0pt,align=center]          at  (2.4,0.6)   (S2)  {S2\\gr};
\node[state,inner sep=0pt,align=center]          at  (3.8,1.0)   (S3)  {S3\\rr};
\node[state,inner sep=0pt,align=center]          at  (5.0,0.6)   (S4)  {S4\\rr};
\node[state,inner sep=0pt,align=center]          at  (6.2,0.3)   (S5)  {S5\\rg};
\node[state,inner sep=0pt,align=center]          at  (6.2,1.5)   (S6)  {S6\\rg};
\node[state,inner sep=0pt,align=center]          at  (7.8,1.5)   (S7)  {S7\\rg};
\node[state,inner sep=0pt,align=center]          at  (7.8,0)     (S8)  {S8\\rr};
\node[state,inner sep=0pt,align=center]          at  (3.8,0)     (S9)  {S9\\rr};

\path
(S0) edge (S1)
(S1) edge (S2)
(S2) edge [loop above,looseness=6] 
     node[xshift=6mm,yshift=-4mm] {$\neg$car} (S2)
(S2) edge               
     node[xshift=4.1mm,yshift=-4mm] {car} (S3)
(S3) edge (S4)
(S4) edge (S5)
(S5) edge 
     node[xshift=-5mm,yshift=-0.7mm] {$\neg$car} (S8)
(S5) edge               
     node[xshift=0mm,yshift=0mm] {car} (S6)
(S6) edge               
     node[xshift=-4.4mm,yshift=1mm] {$\neg$car} (S8)
(S6) edge               
     node[xshift=0mm,yshift=0mm] {car} (S7)
(S6) edge [bend right,dashed,color=darkred] 
     node[xshift=10mm,yshift=1.5mm] {$\neg$car} (S0)
(S7) edge (S8)
(S8) edge [bend angle=12, bend left]   (S9)
(S9) edge [bend angle=12, bend left]   (S0)
;
\end{tikzpicture}

%% file: figs/amba_g3.tex
\begin{tikzpicture}[auto,node distance=\nd]
\node[state,initial above] at  (0,0)       (S0) {S0};
\node[state]               at  (0.6,-1.3)    (S4) {S4};
\node[state]         at  (2.3,0)     (S3) {S3};
\node[state]         at  (4.2,0)     (S2) {S2};
\node[state]         at  (6.1,0)     (S1) {S1};
\node[state]         at  (4.2,-1.3)  (SX) {S$_x$};

\path
(S0) edge [in=100,out=130,loop]  
     node [xshift=12mm,yshift=-1mm] 
     {$\neg (\texttt{B} \wedge \texttt{s})$} (S0)
(S0) edge               
     node [xshift=-7mm,yshift=-1.5mm]  
     {$\texttt{B} \wedge \texttt{s} \wedge \neg\texttt{R}$} (S4)   
(S0) edge               
     node [xshift=0mm,yshift=0mm]  
     {$\texttt{B} \wedge \texttt{s} \wedge \texttt{R}$} (S3)
     
(S4) edge               
     node [rotate=35,xshift=5.7mm,yshift=-4mm]  
     {$\neg \texttt{s} \wedge \texttt{R}$} (S3) 
(S4) edge [in=-40,out=-10,loop]  
     node [xshift=-1mm,yshift=3mm] 
     {$\neg \texttt{s} \wedge \neg \texttt{R}$} (S4)    
(S4) edge  
     node [xshift=0mm,yshift=0mm] 
     {$\texttt{s}$} (SX)
     
(S3) edge               
     node [xshift=0mm,yshift=0mm]  
     {$\neg \texttt{s} \wedge \texttt{R}$} (S2) 
(S3) edge [loop above]  
     node [xshift=8mm,yshift=-5mm] 
     {$\neg \texttt{s} \wedge \neg \texttt{R}$} (S3)
(S3) edge  
     node [xshift=-1mm,yshift=-1mm] 
     {$\texttt{s}$} (SX)
     
(S2) edge               
     node [xshift=0mm,yshift=0mm]  
     {$\neg \texttt{s} \wedge \texttt{R}$} (S1) 
(S2) edge [loop above]  
     node [xshift=8mm,yshift=-5mm] 
     {$\neg \texttt{s} \wedge \neg \texttt{R}$} (S2)
(S2) edge  
     node [xshift=-1mm,yshift=1mm] 
     {$\texttt{s}$} (SX)
     
(S1) edge [in=60, out=90, controls=+(90:1.5) and +(60:1.5)]              
     node [xshift=0mm,yshift=1mm]  
     {$\neg \texttt{s} \wedge \texttt{R}$} (S0) 
(S1) edge [loop below]  
     node [xshift=-3mm,yshift=0mm] 
     {$\neg \texttt{s} \wedge \neg \texttt{R}$} (S1) 
(S1) edge  
     node [xshift=-1mm,yshift=1mm] 
     {$\texttt{s}$} (SX)
     
(SX) edge [in=-10,out=-40,loop]  
     node [xshift=8mm,yshift=-2mm] 
     {$\true$} (SX)  

     
;
\end{tikzpicture}

%% file: figs/unreal_example.tex
\begin{tikzpicture}[auto,node distance=\nd]
\node[state,initial above] at  (0,0)       (r0) {$r_0$};
\node[state]               at  (2,0.8)   (r1) {$r_1$};
\node[state]               at  (2,-0.6)  (r2) {$r_2$};
\node[state]               at  (4.2,0)  (rx) {$r_x$};
\path
(r0) edge               
     node [rotate=21,xshift=7.5mm,yshift=0mm]  {$o_1 \wedge \neg o_2$} (r1)
(r0) edge               
     node [rotate=-20,xshift=-4mm,yshift=-4mm] {$\neg o_1$} (r2)
(r0) edge               
     node [xshift=-3mm,yshift=-1mm] {$o_1 \wedge o_2$} (rx)
(r1) edge [in=30,out=0,loop]              
     node [xshift=15mm,yshift=2mm]  {$o_1 \wedge \neg o_2$} (r1)
(r2) edge [loop right]              
     node [xshift=-5mm,yshift=-3.5mm]  {$\neg o_1$} (r2)
(r2) edge          
     node [rotate=20,xshift=4mm,yshift=-4mm]  {$o_1$} (rx)
(rx) edge [loop below,looseness=6]              
     node [xshift=0mm,yshift=1mm]  {$\true$} (rx)
(r1) edge          
     node [rotate=-18,xshift=-7.0mm,yshift=-0.7mm]  {$\neg o_1 \vee o_2$} (rx)
     
;
\end{tikzpicture}